\documentclass[11pt]{article}

\usepackage{amsmath}
\usepackage{amssymb}
\usepackage{enumerate}
\usepackage{abstract}

\usepackage{amsthm}
\usepackage[margin=1in]{geometry}

\newtheorem{fact}{Fact}
\newtheorem{theorem}{Theorem}
\newtheorem{corollary}{Corollary}
\newtheorem{lemma}{Lemma}

\theoremstyle{definition}

\title{Checking generalized debates with small space and randomness \\}

\author{\\ H. G{\"{o}}kalp Demirci and A. C. Cem Say \\
\\
\small Bo\u{g}azi\c{c}i University, Department of Computer Engineering, Bebek 34342 \.{I}stanbul, Turkey
\\
\small \texttt{\{gokalp.demirci,say\}@boun.edu.tr}
\\ \\
}

\date{}

\begin{document}

\newgeometry{margin=1.2in}

\maketitle

\begin{abstract}
 
We introduce a model of probabilistic debate checking, where a silent resource-bounded verifier reads a dialogue about the membership of the string in the language under consideration between a prover and a refuter. Our model combines and generalizes the concepts of one-way interactive proof systems, games of incomplete information, and probabilistically checkable complete-information debate systems. We consider debates of partial and zero information, where the prover is prevented from seeing some or all of the messages of the refuter, as well as those of complete information. The classes of languages with debates checkable by verifiers operating under severe bounds on the memory and randomness are studied.

We give full characterizations of versions of these classes corresponding to simultaneous bounds of $O(1)$ space and $O(1)$ random bits, and of logarithmic space and polynomial time. It turns out that constant-space verifiers, which can only check complete-information debates for regular languages deterministically, can check for membership in any language in $\mathsf{P}$ when allowed to use a constant number of random bits. Similar increases also occur for zero- and partial- information debates, from $\mathsf{NSPACE}(n)$ to $\mathsf{PSPACE}$, and from $\mathsf{E}$ to $\mathsf{EXPTIME}$, respectively. Adding logarithmic space to these constant-randomness verifiers does not change their power. When logspace debate checkers are restricted to run in polynomial time without a bound on the number of random bits, the class of debatable languages equals $\mathsf{PSPACE}$ for all debate types. 
We also present a result on the hardness of approximating the quantified
max word problem for matrices that is a corollary of this
characterization.
\end{abstract}

\newpage

\restoregeometry

\section{Introduction}
\label{sec:intro}

An alternating Turing machine for a language $L$ can be viewed \cite{Go08} as a system where a deterministic Turing machine (the ``verifier") makes a decision about whether the input string is a member of $L$ upon reading a ``complete-information" debate on this issue between a ``prover" and a ``refuter." Reif \cite{Re79} has shown that the class of languages with debates checkable by space-bounded verifiers is enlarged significantly when the debate format is generalized so that the refuter is allowed to hide some of its messages from the prover: For constant-space verifiers, the class of languages with such partial-information debates is $\mathsf{E}$, whereas the corresponding class for complete-information debates is just the regular languages. Zero-information debates, where the refuter is forced to hide all its messages from the prover, correspond to the class $\mathsf{NSPACE}(n)$ under this resource bound \cite{PR79}.

The case where the verifier is upgraded to a probabilistic Turing machine has first been studied by Condon et al. \cite{CFLS95}, who showed that polynomial-time verifiers that use logarithmically many random bits, and read only a constant number of bits of a complete-information debate can handle every language in $\mathsf{PSPACE}$. Demirci et al. \cite{DSY12} initiated the study of probabilistic verifiers for partial-information debates, and proved some lower-bound results under certain restrictions on the behaviors of the prover and the refuter.

Analyzing the power of computational models under severe resource bounds often gives insight about their nature. In this paper, we characterize the three classes of languages that have complete, zero, and partial-information debates with constant-space verifiers which use only a constant number of random bits, independent of the length of the input, as $\mathsf{P}$, $\mathsf{PSPACE}$, and $\mathsf{EXPTIME}$, respectively.  We also consider verifiers for all these three types of debates with simultaneous logarithmic-space and polynomial-time bounds, without constraining the number of coin tosses, and show that the corresponding classes coincide with $\mathsf{PSPACE}$. Our proofs do not depend on the restrictions of  \cite{DSY12} on the prover and the refuter.

The rest of the paper is structured as follows: Section \ref{sec:preliminaries} describes our debate checking model and relevant previous work. We examine the power of constant-randomness checkers with simultaneous small space bounds for the three debate types mentioned above in Section \ref{sec:consrand}. Section \ref{sec:polytime} contains a characterization of languages with debates checkable by logspace polynomial-time verifiers. A nonapproximability result that follows from our work is presented in Section \ref{sec:nonapp}. Section \ref{sec:conc} is a conclusion.

\section{Preliminaries} 
\label{sec:preliminaries}

\subsection{The model}
\label{subsec:model}

Our formal model of debate systems combines and generalizes the concepts of one-way interactive proof systems \cite{CL89,Co93B}, games of incomplete information \cite{Re79,PR79}, probabilistically checkable complete-information debate systems \cite{CFLS95,CFLS97}, and our previous work on asymmetric debates \cite{DSY12}. The verifier is a probabilistic Turing machine (PTM) with a read-only input tape and a single read/write work tape. The input tape holds the input string between two occurrences of an end-marker symbol, and we assume that the machine's program never attempts to move the input head beyond the end-markers. The input tape head is on the left end-marker at the start of the process. The verifier reads information written by the prover and the refuter, (conventionally named Player 1 and Player 0, and denoted P1 and P0, respectively,) by alternately consulting two \textit{reading cells}, C1 and C0. When it is consulted for the $i$th time, cell C1 returns the $i$th symbol in P1's argument for the membership of the input string in the language under consideration. C0 similarly yields the next symbol in P0's counterargument. The symbols used in this communication are chosen from three   alphabets named $\Gamma_1$, $\Gamma_0$, and $\Delta$, such that $\Gamma_0 \cap \Delta = \emptyset$. For $i \in \{0,1\}$, $\Gamma_i$ is the set of symbols emittable by P$i$ that can be seen by the opposing player. $\Delta$ is the set of private symbols that P0 may choose to write to the verifier without showing to P1. P0 can use any member of $\Gamma_0 \cup \Delta $ in its messages. P1 (resp., P0) is assumed to have seen the subsequence of all the public symbols (i.e. those in $\Gamma_0$ and $\Gamma_1$) among the first $i-1$ (resp., $i$) symbols emitted up to that time before preparing its $i$th symbol. 
Using this communication infrastructure, P1 attempts to convince the verifier that it is right, and P0 tries to prove P1 wrong. We allow the possibility that the cheating player sends an infinite sequence of symbols, in order to try to make the verifier run forever, rather than to arrive at the correct decision. The verifier also has access to a source of private random bits. The state set of the verifier TM is $Q$, containing, among others, two special halting states, $q_{a}$ (accept) and $q_{r}$ (reject). One of the non-halting states is designated as the start state. 

Two (possibly different) subsets $R$ and $T$ of $Q$ are designated as the sets of \textit{reading} and \textit{coin-tossing states}, respectively. $R$ is partitioned to two subsets, $R_1$ and $R_0$, where the program of the machine (to be described shortly) ensures that the first reading state to be entered is in $R_1$, and the sequence of reading states entered during computation alternates between members of $R_1$ and $R_0$. Whenever a state $q \in R_i$ is entered, the next symbol of P$i$'s message is written in reading cell C$i$. Let $\lozenge=\{-1,0,+1\}$ denote the set of possible head movement directions. The program of the verifier is formalized by the  transition function $\delta$ as follows: For $i \in \{0,1\}$ and $q \in R_i\cap T -\{q_{a},q_{r}\}$, $\delta(q,\zeta,\theta,\sigma,b)=(q',\theta',d_{ih},d_{wh})$ means that the machine will switch to state $q'$, write $\theta'$ on the work tape, move the input head in direction $d_{ih} \in \lozenge$, and the work tape head in direction $d_{wh} \in \lozenge$, if it is originally in state $q$, scanning the symbols   $\zeta$, $\theta$, and $\sigma$ in the input and work tapes, and the reading cell C$i$, respectively, and seeing the random bit $b$ as a result of the coin toss. If $q \in T - (R \cup \{q_{a},q_{r}\})$, a restricted version of the verifier transition function described above that does not involve the symbol from the reading cell is applied. For $q \in Q-T$, the program's access to the reading cells is again determined by whether $q$ is in $R_0$, $R_1$, or $Q-(T\cup R)$ as described above, but with restricted transitions that do not use a random bit.

A \textit{configuration} of the verifier is a 4-tuple containing its state, head positions, and work tape contents, as usual. Configurations whose state components are in the set $R$ will be called \textit{reading configurations}.

We will now describe an infinite tree, called the \textit{debate tree}.  The even-numbered levels (including level 0, containing the root node) of the debate tree will correspond to points in the debate where it is P1's turn to speak, so nodes at those levels are called P1 nodes. The remaining levels correspond to P0. Each P1 node has $|\Gamma_1|$ children, with each  one of the edges connecting it to its children corresponding to a different symbol that P1 can emit at that point in the debate. Each P0 node has $|\Gamma_0|+|\Delta|$ children, with a similar meaning. The edges connecting a P0 node to its children are called  \textit{P0 edges}.

At any point in the debate, P1 can base its decision on what to say next on the sequence of P0 symbols that it has ``seen" up to that moment. 
 Let $h: (\Gamma_0 \cup \Delta) \rightarrow (\Gamma_0 \cup \{\flat\})$ be a function such that $h(\sigma)=\sigma$ for all $\sigma \in \Gamma_0$, and $h(\sigma)=\flat$ for all $\sigma \in \Delta$, where $\flat$ is a ``blank" symbol, not in $\Gamma_0 \cup \Delta$. 
 For any P1 node $N$, the \textit{P0 sequence seen by N} is created by starting with the empty list, and adding $h(\sigma)$ for each symbol $\sigma$ one encounters at P0 edges while  walking from the root node down to $N$.

A \textit{debate subtree} is a subtree of the debate tree where each P1 node has just one child, and each P0 node has $|\Gamma_0|+|\Delta|$ children. A debate subtree is said to be \textit{well-formed} if  two P1 nodes, say, $N_1$ and $N_2$, 
at the same level emit different symbols only if the P0 sequences seen by $N_1$ and $N_2$ are different. Intuitively, this well-formedness condition corresponds to our desire that P1's messages should reflect the level of ignorance that it has about the private messages of P0. 

Formally, a \textit{debate} is a sequence of symbols that labels an infinite path starting at the root node of a well-formed debate subtree.  The general definition we have given corresponds to \textit{partial-information} debates. When $\Delta=\emptyset$, (i.e. P0 never emits private symbols), one has a \textit{complete-information} debate. The other extreme, where P0 never emits public symbols ($\Gamma_0=\emptyset$), corresponds to \textit{zero-information} debates.

We will associate languages with debate systems in two different ways. We start with the ``strong" definition.

The probability that a verifier $V$ accepts an input string $w$ (i.e. ends up in $q_{a}$) as the result of watching a debate $\pi$ between P1 and P0 presented to it through the reading cells as described above is denoted by $P(a)^w_{(V,\pi)}$. $P(r)^w_{(V,\pi)}$ denotes the probability that $V$ rejects $w$ in such a scenario.

We say that language $L$ \textit{has a debate checkable with error probability} $\varepsilon$ if there exists a verifier $V$ such that
\begin{enumerate}
	\item for every $w \in L$, there is a well-formed debate subtree on which, for all debates $\pi$ labeling a path of this subtree, $P(a)^w_{(V,\pi)} \geq 1-\varepsilon$,  and,
	\item for every $w \notin L$, on all well-formed debate subtrees, there exists a debate $\pi$ labeling some path of the subtree such that $P(r)^w_{(V,\pi)} \geq 1-\varepsilon$.
\end{enumerate}

\indent When we replace item (2) above with the following condition, we obtain the ``weak definition" of debate checking: 
\begin{enumerate}
	\item [2]$ \mspace{-8mu}^{'} \mspace{-5mu}. $ for every $w \notin L$, on all well-formed debate subtrees, there exists a debate $\pi$ labeling some path of the subtree such that $P(a)^w_{(V,\pi)} \leq \varepsilon$.
	
\end{enumerate}

$\mathsf{CDEB}(s,t,r)$ is the class of languages that have complete-information debates checkable with some error probability $\varepsilon<\frac{1}{2}$, such that the verifier of the debate uses $s(n)$ space, $t(n)$ time, and $r(n)$ random bits. Classes of languages that have partial-information debates and zero-information debates with these restrictions have names of the form $\mathsf{PDEB}(s,t,r)$ and $\mathsf{ZDEB}(s,t,r)$, respectively. $\mathsf{CDEB}_w(s,t,r)$, $\mathsf{PDEB}_w(s,t,r)$, $\mathsf{ZDEB}_w(s,t,r)$ are the corresponding classes of languages recognized according to weak definition. Note that, since these ``weak" debate systems are less constrained than the ``strong" ones of the previous definition, the $\mathsf{xDEB}$ classes are always contained in the corresponding $\mathsf{xDEB}_w$ classes for $\mathsf{x \in \{C,Z,P\}}$. We use notations $cons$, $log$, $poly$, $exp$ to stand for functions in $O(1)$, $O(\log n)$, $O(n^c)$, $O(2^{n^c})$, respectively, for any constant $c$. If there is no restriction on a resource, we indicate this by $\infty$. When the randomness parameter is set to 0, this indicates that the verifier is deterministic.

We differentiate between the class of languages that have debates checkable with some error probability, and the class of languages that have debates checkable for \textit{all} positive error probabilities $\varepsilon < \frac{1}{2}$. The latter is denoted by adding an $^*$ at the end of the corresponding class name (e.g. $\mathsf{PDEB}^*(s,t,r)$ is the class of languages that have partial information debates for all positive error bounds $\varepsilon <\frac{1}{2} $). The $\mathsf{xDEB^*}$ classes are of course always contained in the corresponding $\mathsf{xDEB}$ classes.


\subsection{Relation to previous work}
\label{subsec:alternation}

The alternating Turing machines (ATMs) of \cite{CKS81} correspond precisely to deterministic verifiers reading complete-information debates in our terminology. The state set of an ATM is partitioned to sets of existential and universal states. In the parlance of Section \ref{subsec:model}, the existential states correspond to the set $R_1$, and the symbol sent by P1 and read from C1 corresponds to the ``existential choice" made by the ATM at that step. Universal choices correspond to P0 symbols in a similar way. Some well-known facts regarding ATMs can be phrased in our setup as follows:


\begin{fact} \label{fact:ASPACE}
$\mathsf{CDEB}(s(n),\infty,0)=\mathsf{ASPACE}(s(n)) = \mathsf{DTIME}(2^{O(s(n))})$ for any $s(n) \geq \log n$. \cite{CKS81}
\end{fact}


\begin{fact}
$\mathsf{CDEB}(\infty,t(n),0)=\mathsf{ATIME}(t(n)) = \mathsf{DSPACE}(t(n))$ for any $t(n) \geq n$. \cite{CKS81}
\end{fact}

Reif \cite{Re79} introduced a natural extension to alternation by giving computational models for games of incomplete information, where the existential player does not have access to all of the opponent's moves. To incorporate this notion to the alternating Turing machine framework, Reif augmented the ATM model with \textit{private work tapes} and \textit{private states} in addition to the usual tapes and states that are common to both ``players." In a \textit{private alternating Turing machine} (PATM), only the universal states can execute moves that can see or change the content of the private tapes and states. This prevents the strategy of the existential player from depending on the content of private tapes and states, effectively enforcing the same condition mentioned in our definition of well-formed debate subtrees on the existential choices. When the  universal states of a PATM are forbidden to make moves that change the common memory elements, we get a \textit{blind alternating Turing machine} (BATM). The equivalence of these two models to our setup with deterministic verifiers reading partial and zero-information debates, respectively, is demonstrated in the Appendix. We know the following facts about the language recognition power of PATMs and BATMs. 

\begin{fact}
\begin{tabular}{@{}l@{}}
$\mathsf{PDEB}(s(n),\infty,0)=\mathsf{PASPACE}(s(n)) = \mathsf{DTIME}(2^{2^{O(s(n))}})$, \\
$\mathsf{ZDEB}(s(n),\infty,0)=\mathsf{BASPACE}(s(n)) = \mathsf{DSPACE}(2^{O(s(n))})$ for any $s(n) \geq \log n$. \cite{Re79}
\end{tabular}
\end{fact} 

\begin{fact}\label{fact:patime}
\begin{tabular}{@{}l@{}}
$\mathsf{PDEB}(\infty,t(n),0)=\mathsf{PATIME}(t(n)) = $ \\ 
$\mathsf{ZDEB}(\infty,t(n),0)=\mathsf{BATIME}(t(n)) = \mathsf{DSPACE}(t(n))$ for any $t(n) \geq n$. \cite{Re79}
\end{tabular}
\end{fact} 

An interactive proof system (IPS) can be seen as a variant of our model where P0 says nothing, and the verifier is allowed to exchange messages with P1. An IPS with perfect completeness, that is, a guarantee that members of the language are accepted with probability 1, can be transformed to a PATM with the same space and time bounds by just viewing the the coin-tosses of the verifier as private universal moves, and the branchings due to the prover messages as possible existential moves of the PATM. By combining the facts that every IPS can be assumed to have perfect completeness \cite{FGMSZ89}, and that $\mathsf{PSPACE}$ equals the class of languages that have IPSs with polynomial-time logspace verifiers \cite{Co91,Sh92} with Fact \ref{fact:patime}, we conclude the following about simultaneously space and time bounded PATMs:
\begin{fact}\label{fact:pdeblogpoly0}
$\mathsf{PDEB}(log,poly,0)=\mathsf{PSPACE}$. 
\end{fact}

Constant-space versions of alternating machines have also been studied: ($\mathsf{REG}$ denotes the class of regular languages.)

\begin{fact}
\begin{tabular}{@{}l@{}}
$\mathsf{CDEB}(cons, \infty ,0)= \mathsf{REG}$, \cite{LLS78}\\ 
$\mathsf{ZDEB}(cons, \infty ,0)= \mathsf{NSPACE}(n)$, \cite{PR79} \\
$\mathsf{PDEB}(cons, \infty ,0)= \mathsf{E}$. \cite{PR79} \\
\end{tabular}
\end{fact}

It is well known that multiple input heads are equivalent to logarithmic space \cite{Ha72}, for instance, the class of languages recognized by deterministic multihead finite automata is precisely $\mathsf{L}$. This equivalence carries over to the nondeterministic, probabilistic \cite{Ma97}, and alternating versions of these machines: Alternating multihead finite automata with two-way access to the input (denoted 2afa($k$) for $k$ input heads) are known to characterize alternating logspace: ($\mathsf{2AFA}(k)$ denotes the class of languages recognized by 2afa($k$)s.)

\begin{fact}\label{fact:2afakp}
$\cup_{k \geq 1}\mathsf{2AFA}(k) = \mathsf{P}$. \cite{Ki88}
\end{fact}


When the input heads are defined as resources private to the universal player, and state information is used for public communications, it can be shown, again with the technique of \cite{Ha72}, that blind and private alternating multihead finite automata (2bafa($k$)s and 2pafa($k$)s, respectively,) are equivalent to the corresponding logspace models. We name the $k$-head classes associated with these machines $\mathsf{2BAFA}(k)$ and $\mathsf{2PAFA}(k)$, respectively.

\begin{fact}
\begin{tabular}{@{}l@{}}
$\cup_{k \geq 1}\mathsf{2BAFA}(k) = \mathsf{BASPACE}(log)=\mathsf{PSPACE}$, \\
$\cup_{k \geq 1}\mathsf{2PAFA}(k) = \mathsf{PASPACE}(log)=\mathsf{EXPTIME}$.
\end{tabular}
\end{fact}  


\section{Constant randomness}
\label{sec:consrand}
Condon et al. \cite{CFLS95} initiated the study of probabilistic debate checking with their work on complete-information debates. They only considered polynomial-time and logarithmic-randomness bounds on their verifiers. In this section, we characterize the languages that have debates checkable by verifiers that are allowed to use only constant amounts of memory and randomness, not just for complete-, but also for zero- and partial-information debates. Our results here are obtained by adapting a technique discovered by 
Say and Yakary{\i}lmaz \cite{SY12} for simulating nondeterministic logarithmic-space machines by verifiers under these strict resource bounds. Note that one other difference between the model of \cite{CFLS95} and ours is that our verifiers do not have the capability of accessing desired locations of the debate directly, and must read the debate sequentially until they arrive at a decision.


\begin{lemma}
\label{PinCDEBw}
$\mathsf{P} \subseteq \mathsf{CDEB}^*_w(cons, \infty , cons)$.
\end{lemma}

\begin{proof}
Let $L$ be any language in $\mathsf{P}$, and let $M$ be the 2afa($k$) recognizing $L$  by Fact \ref{fact:2afakp}. We will construct a constant-space verifier which tosses only a constant number of coins while checking complete-information debates about membership in $L$ with  a desired positive error probability $\varepsilon$ according to the weak definition.

Without loss of generality, we make certain assumptions about $M$: All states of $M$ are either existential or universal. The starting and the halting states of $M$ are existential. Each existential and universal state leads exactly to two different branches, and computation alternates between existential and universal states.

In the debate, P1 and P0 are supposed to be talking about the step-by-step execution of $M$ on the input string $w$ in a sequence of exchanges, each of which treats a pair of $M$-steps. Each exchange has the following structure:
\begin{itemize}
\item P1 announces its claim about which existential choice to be made by $M$ at the present step guarantees reaching an accept state eventually. P1 also gives a list of $k$ tape symbols of $M$,  claiming that the $i$th head of $M$ will be scanning the $i$th symbol in this list after taking the step just announced by P1, for all $i$, where $1 \leq i \leq k$.
\item P0 announces the universal choice of $M$ for the next step that it claims will lead to a rejection of $w$.
\item P1 announces the list of $k$ tape symbols which it claims that will be scanned by $M$'s heads after the execution of the universal step according to the choice just specified by P0.
\end{itemize}
(Our definition in Section \ref{sec:preliminaries} stipulates that P1 and P0 should alternate after emitting each symbol. The protocol here can be made to fit that requirement, since the  ``packages" sent by P1 are of fixed length, and P0 can emit dummy symbols while listening to them. The verifier can measure the length of the packages using its constant-sized memory, and reject the input if it sees P1 giving a package of the wrong length.)

If P1 claims that the set of moves described up to that point in the debate will have caused $M$ to halt with acceptance, it emits the special symbol $\circlearrowleft$ indicating this claim, and restarts the procedure by giving the first existential choice from the initial configuration of $M$ on $w$.

We now describe the verifier  $V$. $V$ will read the debate to simulate the execution of $M$ on the particular computation path indicated by the alternating messages of P1 and P0. It starts by initializing a counter to 1, and randomly selecting one of the $k$ heads of $M$, using $r=\left\lceil \log k \right\rceil$ random bits. Each head has a probability of at least $2^{-r}$ to be selected. $V$ uses its single input head to track the position of the selected head during its simulation of $M$. $V$ needs to know what the other heads are scanning in order  to compute $M$'s next state at any step, and it depends on P1 for this information. After each simulated step, $V$ checks if the symbol scanned by its input head is consistent with what P1 claims about the corresponding head of $M$. If it sees that P1 is lying in this regard, $V$ rejects the input. Otherwise, it continues by reading the next exchange. If $V$ sees $M$ reaching a reject state during the simulation, it rejects. If P1 emits a $\circlearrowleft$, $V$ checks if its simulation of $M$ has indeed reached an accept state. If $V$ sees that $M$ has not accepted despite P1 announcing that it has, it rejects. If $M$ is seen to have accepted, $V$ increments the counter. If the counter has exceeded $c$, whose value will be discussed below, $V$ accepts. Otherwise, $V$ moves its head to its original position, randomly picks a head of $M$ to track, and starts processing the restarted debate. 

If $w \in L$, P1 can always find a way of responding to P0 that will lead $M$ to acceptance, while giving correct information about the symbols scanned by the heads, so $V$ will accept with probability 1. If $w \notin L$, P1 must lie about at least one head at some point to prevent $V$ from rejecting. This lie will be caught, and $V$ will therefore reject, with probability at least $2^{-r}$, in every simulated computation of $M$. The probability that $V$ will fail to pick the head that P1 is lying about in all $c$ iterations of the loop is $(1-2^{-r})^c$, which can be tuned to be below $\varepsilon$ by choosing a sufficiently large $c$.

An evil P1 can cause $V$ to spend infinite time without ever restarting the computation of $M$ with probability at most $(1-2^{-r})$. We conclude that $V$ is a weak debate checker for $L$ with the desired properties.
\end{proof}
%



\begin{lemma}
\label{PinCDEB}
$\mathsf{P} \subseteq \mathsf{CDEB}(cons, \infty , cons)$.
\end{lemma}

\begin{proof}
Consider the verifier $V$ from the proof of Lemma 
\ref{PinCDEBw} with the number of iterations $c$ set to 1. 
Recall that this $V$ uses $r=\left\lceil \log k \right\rceil$ random bits to simulate one run of the 2afa($k$) that recognizes a language $L$ in $\mathsf{P}$. $V$ accepts input strings in the language with probability 1, and rejects strings not in the language with probability $2^{-r}$.

To increase the rejection probability for nonmembers so that the resulting machine fits the definition in Section \ref{sec:preliminaries}, we construct a new verifier $V'$, which starts by using an additional $r+1$ random bits to  reject the input directly with probability $\frac{2^r -1}{2^{r+1}}$, and transfers control to $V$ with the  remaining probability. $V'$ can be seen to accept inputs in $L$ with probability at least $1-\left( \frac{2^r -1}{2^{r+1}} \right)=\frac{2^r +1}{2^{r+1}}$, and rejects the inputs not in $L$ with probability $2^{-r} \left( \frac{2^r +1}{2^{r+1}} \right) + \frac{2^r -1}{2^{r+1}} = \frac{2^{2r} +1}{2^{2r+1}} $. Thus, $V'$ constitutes a debate checker for $L$ according to the strong definition with error bound $\frac{2^{2r} -1}{2^{2r+1}}$.     
\end{proof}




We now derive some time bounds that will be useful in our characterization theorems.

\begin{lemma}\label{CDEBwinCDEB}
\begin{tabular}{@{}l@{}}
${\mathsf{CDEB}}_w(log, \infty , cons) \subseteq {\mathsf{CDEB}}(log, poly, cons)$, \\
${\mathsf{ZDEB}}_w(log, \infty , cons) \subseteq {\mathsf{ZDEB}}(log, exp, cons)$, \\
${\mathsf{PDEB}}_w(log, \infty , cons) \subseteq {\mathsf{PDEB}}(log, exp, cons)$.
\end{tabular}
\end{lemma}

\begin{proof}
We start by showing that a complete-information debate checkable by a verifier using logarithmic space and a constant amount of random bits does not need to be longer than some polynomially bounded number of symbols. 

Let $V$ be the logspace verifier for a language $L \in {\mathsf{CDEB}}_w(log, \infty , cons)$ that uses at most $r$ random bits for some constant $r$. $V$ can be derandomized as a collection of $2^r$ deterministic logspace verifiers S=$\{V_1, V_2,.. , V_{2^r}\} $ for each different assignment to the random sequence of length $r$, such that, for inputs in $L$, the common debate that the verifiers in S are reading eventually makes more than half of them accept. Each $V_i$ can be in one of polynomially many different possible configurations. Let E be the set of all possible combined states of the ``ensemble" consisting of these verifiers. The cardinality $C$ of E  is itself polynomially bounded.


For any string $w \in L$, we know that there is a debate subtree, i.e. a ``best strategy" for P1, where P1 ``wins" in every branch by steering $V$ toward an ensemble with a majority of accepting configurations. If a debate corresponding to a particular branch of that tree has been going on for more than $C$ turns without reaching such an ``accepting" ensemble, this means that P1 does not have a sequence of clever responses to prevent P0 from causing a loop of ensemble states, and so this debate will never end up with an accepting ensemble. This contradiction leads us to conclude that no complete-information debate that is checkable by a constant-coin logspace verifier needs to have superpolynomial length. We construct a new verifier that simulates $V$, while using its logarithmic memory to clock this simulation, cutting off and rejecting when either $V$ is seen to enter an infinite loop without reading debate symbols, or when the debate goes on for too long.

For partial-information debates, we modify the argument above to take the increased ignorance of P1 about the state of $V$ into account. Since P1 does not know what P0 has been saying to $V$, it does not know precisely what ensemble state $V$ is in at any point during the debate. From P1's point of view, $V$ can be in this or that ensemble state, having received this or that message from P0, so P1's strategy has to be based on viewing $V$ as being in a \textit{set} of possible ensemble states that are consistent with what little P1 knows about P0's messages up to that point \cite{PR79}, and finding an argument that would lead all those ensembles to acceptance. 
We now see that any partial-information debate longer than the cardinality of the power set of E must involve a repetition of a set of ensembles, which means a failure for P1. By the same reasoning as above, no partial-information debate that is checkable by a constant-coin logspace verifier needs to have superexponential length.


Blind alternating Turing machines can count up to $2^{2^{s(n)}}$ using $s(n)$ space for all $s(n) \geq \log n$ (see Theorem 1 in \cite{PR79}). We can integrate this counting mechanism to any logspace constant-coin verifier to reject the input if the partial-information debate that it has been reading has exceeded the exponential time bound derived above.
\end{proof}

\begin{lemma}
\label{CDEBwinP}
$\mathsf{CDEB}(log, poly , cons) \subseteq \mathsf{P}$.
\end{lemma}

\begin{proof}
Let $L$ be a language in $\mathsf{CDEB}(log, poly , cons)$. Let $V$ be the polynomial-time, logspace verifier that checks complete-information debates about $L$ with bounded error using $r$ random bits.
Without loss of generality, we assume that each configuration of $V$ that reads C1 is immediately followed by a configuration that reads C0. We construct an alternating logspace Turing machine $M$ that recognizes $L$. 


As in the proof of Lemma \ref{CDEBwinCDEB}, we think of $V$ as a collection of $2^r$ deterministic logspace verifiers S=$\{V_1, V_2,.. , V_{2^r}\}$. $M$ simulates the elements of S in a time-sharing fashion. It first advances each $V_i$  until they reach their first reading configuration, or halt. When all the $V_i$ that have not halted yet are ready to consume a prover symbol, $M$ branches existentially to produce a symbol  from $\Gamma_1$, and feeds this symbol to all those $V_i$s. $M$ then branches universally to produce a symbol from $\Gamma_0$ to feed to the verifiers waiting to read C0. $M$ keeps simulating each $V_i$ in this fashion by feeding them alternately produced symbols as the next pair of messages of the debate. Since all the $V_i$s use logarithmic space, the simulation also takes logarithmic space. $M$ accepts if more than half of the deterministic verifiers in the set S has accepted. Otherwise, $M$ rejects the input string.   
\end{proof}

We have proven
\begin{theorem}
\label{PeqCDEB}
$\mathsf{CDEB}(cons, \infty , cons) = \mathsf{CDEB}(log, poly , cons) = \mathsf{P}$.
\end{theorem}

\begin{lemma}
\label{PSPACEinZDEB}
$ \mathsf{PSPACE} \subseteq \mathsf{ZDEB}(cons, \infty , cons)$.
\end{lemma}

\begin{proof}

Let $L$ be any language in $\mathsf{PSPACE}$, and let $M$ be the 2bafa($k$) recognizing $L$. We will first modify the construction in the proof of Lemma \ref{PinCDEBw} to build a verifier $V$ which rejects strings not in $L$ with probability 1, and accepts the members of $L$ with probability at least $2^{-\left\lceil \log k \right\rceil}$. We will then convert $V$ to a debate checker that fits the strong definition, in a manner similar to what we did in Lemma \ref{PinCDEB}.

 $V$ initially picks one of $M$'s heads randomly, and simulates $M$ according to existential and universal choices suggested by P1 and P0, respectively. In Lemma \ref{PinCDEBw}, P1 was the player who is supposed to send the symbols scanned by the heads of the multihead automaton that is being simulated. This is not acceptable in this case, since the configuration of $M$ during its execution depends on choices made by its universal player, and Player 1 in the zero-information protocol that we are designing cannot have that information. P0, on the other hand, can see the full configuration of $M$, just as the universal player of a 2bafa($k$), and it is therefore P0 who provides the symbols scanned by the heads of $M$. $V$ accepts the input if it reaches an accept state of $M$, or detects an inconsistency between its own input and the claim made for the head that it is tracking  by P0. Otherwise, $V$ rejects.

Therefore, $V$ accepts inputs in $L$ with probability at least $2^{-r}$, using $r=\left\lceil \log k \right\rceil$ random bits. It rejects inputs not in $L$ with probability 1. We construct a new verifier $V'$ that accepts directly with probability $\frac{2^r -1}{2^{r+1}}$,  and transfers control to $V$ with the remaining probability. Then, $V'$ accepts inputs in $L$ with probability at least $\frac{2^{2r} +1}{2^{2r+1}} $, and rejects the inputs not in $L$ with probability $\frac{2^r +1}{2^{r+1}}$. $V'$ is therefore a strong zero-information debate checker for $L$ with error bound $\frac{2^{2r} -1}{2^{2r+1}}$.  
\end{proof}

\begin{lemma}
\label{ZDEBwinPSPACEalt}
${\mathsf{ZDEB}}(log, exp , cons)  \subseteq  \mathsf{PSPACE} $.
\end{lemma}

\begin{proof}
Let $L$ be any language in $\mathsf{ZDEB}(log, exp, cons)$, and let $V$ be the constant-coin, exponential-time logspace verifier of zero-information debates on $L$. 
We build a logspace blind alternating Turing machine $M$ recognizing $L$. The construction is almost identical to that of the proof of Lemma \ref{CDEBwinP}. Since the debates of $V$ are zero-information, $M$ simulates the $2^r$ deterministic verifiers (the $V_i$s) on the private tape of the universal player, feeding them symbols from the private alphabet $\Delta$, ensuring that the moves of the existential player obey the zero-information condition. Logaritmic space is sufficient for this simulation. As in Lemma \ref{CDEBwinP}, $M$ accepts if and only if more than half of the $V_i$s accept, which happens if and only if the input is in $L$.
\end{proof}

Lemmas \ref{PSPACEinZDEB} and \ref{ZDEBwinPSPACEalt} form
\begin{theorem}
\label{PSPACEeqZDEB}
$\mathsf{ZDEB}(cons, \infty , cons) = \mathsf{ZDEB}(log, exp , cons) = \mathsf{PSPACE}$.
\end{theorem}

We conclude this section with a characterization for partial-information debates.
\begin{theorem}
\label{EXPeqPDEB}
$\mathsf{PDEB}(cons, \infty , cons) = \mathsf{PDEB}(log, exp , cons) = \mathsf{EXPTIME}$.
\end{theorem}

\begin{proof}
We will describe the necessary modifications to the proofs of Lemma \ref{PSPACEinZDEB} and \ref{ZDEBwinPSPACEalt}.

First, we can design partial-information debates with constant-space verifiers allowed to use only a constant number of random bits to simulate the 2pafa($k$) of any given language in $\mathsf{EXPTIME}$, thereby showing $\mathsf{EXPTIME} \subseteq \mathsf{PDEB}(cons, \infty , cons)$. The only difference with the construction in Lemma \ref{PSPACEinZDEB} is that P0 is allowed to use the public alphabet $\Gamma_0$ as well as the  private alphabet $\Delta$ for suggesting the public and private moves of the universal player of the simulated 2pafa($k$).

Second, simulation of exponential-time constant-randomness logspace partial-information debate systems by logspace private alternating Turing machines implies $\mathsf{PDEB}(log, exp , cons) \subseteq \mathsf{EXPTIME}$. The simulation is similar to the one in Lemma \ref{ZDEBwinPSPACEalt}, but the private alternating TM now produces symbols by universal branching from the set $\Gamma_0 \cup \Delta$ when it needs the next message of P0 in the debate. If the produced symbol is from $\Gamma_0$, the PATM lets the existential player know this symbol by writing  it in a special memory cell on the common work tape. The task of simulating the finitely many exponential-time deterministic verifiers is completed in exponential time.  
\end{proof}

\section{Logarithmic space and polynomial time}
\label{sec:polytime}

Condon \cite{Co93B} has shown that languages recognized by one-way interactive proof systems\footnote{One-way interactive proof systems are just debate checkers which listen only to P1. See, for instance, \cite{SY12} for a review.} with simultaneously polynomial-time and logarithmic-space bounded verifiers are polynomial time reducible to the max word problem for matrices, which is a variation of the well-known word problem for matrices.  We use  her technique below to show that all languages which have partial-information debates checkable by similarly bounded verifiers are in $\mathsf{PSPACE}$.

\begin{lemma}
\label{PDEB(log,poly)inPSPACE}
$\mathsf{PDEB}(log, poly, \infty) \subseteq \mathsf{PSPACE}$.
\end{lemma}

\begin{proof}
Let  $V$ be a logspace verifier that checks a partial-information debate for a language $L$ with error probability $\epsilon$, halting after reading $2t$ symbols of the debate for some polynomial $t$ in the input length. Let $\Gamma_1$, $\Gamma_0$, and $\Delta$ denote the public and private alphabets of the players as usual. Without loss of generality, assume that the initial and halting states of $V$ are in the set $R_1$ (associated with reading cell C1). Since $V$ uses logarithmic space, we can assume that it has $2m$ reading configurations, where $m$ is a polynomial in the input length. Order these configurations so that the first $m$ have their state components in $R_1$, whereas the ones from positions $m+1$ to $2m$ have their state components in $R_0$. Make sure that the initial configuration is at position 1. We will build a polynomial-time private alternating Turing machine $M$ for $L$.

On a specific input string $w$, we define $p(i,j, \sigma)$ as the probability of $V$ eventually reaching reading configuration $j$ (without visiting any other reading configurations in between) from reading configuration $i$ where it reads the symbol $\sigma$ in the corresponding reading cell. Since the computation of $V$ alternates between reading C1 and C0, $p(i,j, \sigma)=0$ for both $i,j \leq m$, and $i,j > m$. Furthermore, the value of any $p(i,j, \sigma)$ depends only on $w$, $i$, $j$, $\sigma$, and the transition function of $V$, and can be computed in polynomial time, using the procedure explained in detail in the proof of  Theorem 2.1 in \cite{Co93B}\footnote{Note that what we call the ``reading configurations" are named ``communication configurations" in \cite{Co93B}.}. 
We define two sets  $W_{P0} =\left\{ W_{0,\sigma} \text{ }|\text{ } \sigma \in \Gamma_0 \cup \Delta  \right\} $ and $W_{P1} = \left\{ W_{1,\sigma} \text{ }|\text{ } \sigma \in \Gamma_1 \right\} $, where each $W_{0,\sigma}$ is an   $m \times m$  matrix containing $p(i+m,j,\sigma)$ as the $j$th entry of its $i$th row, and the $W_{1,\sigma}$s are matrices with $p(i,j+m,\sigma)$ as the $j$th entry of the $i$th row, for $1\leq i,j \leq m$.

$M$ prepares these two sets of matrices, and then starts picking matrices from $W_{P1}$ and $W_{P0}$ via existential and universal moves, respectively. Universal choices are made on the private work tape of $M$. If the universally chosen matrix corresponds to a symbol $\sigma$ from $\Gamma_0$, $M$ lets the existential player know about the decision by writing $\sigma$ on the public work tape. $M$ continues this process until the existential and universal players have picked $t$ matrices each, at which point $M$ calculates the product $W$ of these matrices in the order they were chosen. Let $v$ and $f$ be two vectors with $m$ entries. $v$ has a 1 in the first position, and 0 everywhere else. $f$ has 1's in the positions corresponding to accepting configurations according to the ordering we defined above, and 0 everywhere else. $M$ accepts if the product $vWb^T$ is greater than $1-\epsilon$. Otherwise, it rejects. 

We have seen to it that the $j$th entry of $vW$ is the probability that $V$ reaches the $j$th reading configuration after reading the partial-information debate corresponding to the moves of $M$, and that $M$ accepts if and only if the overall accepting probability of $V$ in this case is sufficiently high. We conclude that $M$ recognizes $L$.
\end{proof}

Lemma \ref{PDEB(log,poly)inPSPACE} and Fact \ref{fact:pdeblogpoly0} show that no amount of randomness can help a polynomial-time logspace verifier reading partial-information debates to ``break the $\mathsf{PSPACE}$ barrier." We now show that this is different for complete-information debates. Recall that $\mathsf{CDEB}(log, \infty,0)=\mathsf{P}$.

\begin{lemma}
\label{PSPACEinCDEB(log,poly)}
\begin{tabular}{@{}l@{}}
$\mathsf{PSPACE} \subseteq \mathsf{CDEB^*}(log, poly,poly)$, \\
$\mathsf{PSPACE} \subseteq \mathsf{ZDEB^*}(log, poly,poly)$.
\end{tabular}
\end{lemma}

\begin{proof}
Let $L$ be a language in $\mathsf{PSPACE}$, and let $M$ be the polynomial-time alternating Turing machine recognizing $L$. Let $t$ be the time bound of $M$ for some polynomial $t$ in the input length. Without loss of generality, we assume that all configurations of $M$ are either existential or universal, the initial and halting configurations are existential, and any computation of $M$ always alternates between existential and universal configurations. We build a polynomial-time logspace verifier $V$ that checks complete-information debates about whether  $M$ accepts the input string. Let $\varepsilon$ be the desired error bound of $V$.

As we had in the protocol described in the proof of Lemma \ref{PSPACEinZDEB}, P1 and P0 are supposed to provide the existential and universal choices of the simulated machine $M$ to the verifier. Since $V$ does not have the resources to store a configuration of $M$, P0 is also expected to be giving a description of the configuration of $M$ after each simulated step. $V$ requires each such configuration description to be exactly $t$ symbols long, so P0 pads these messages with blanks when necessary. The players are supposed to restart the procedure after precisely $t$ configuration descriptions have been transmitted to $V$.


Using logarithmic space, $V$ can easily check that the number and the length of the configurations presented by P0 are legal. It is also easy for $V$ to check if the first configuration message matches the initial configuration of $M$. However, checking whether the present configuration sent by P0 follows from the previous one is not something that $V$ can do deterministically with this little memory. Instead, $V$ randomly picks an integer $k$ ($1 \leq k \leq t-2)$ at the beginning of each simulation of $M$, and compares only the $k$th, $(k+1)$th and $(k+2)$th symbols of the $i$th configuration with the corresponding symbols of the $(i+1)$th configuration,  for $1 \leq i < t-1$. If it sees a violation of the transition rules of $M$ within this window, $V$ detects P0's lie, and accepts the input. If $V$ fails to find any such error by P0 in $d=\lceil \ln \frac{1}{\varepsilon} \rceil t$ successive simulations of $M$, all of which are accompanied with computation paths ending with rejecting configurations, it rejects.

If the input string is not in $L$, there is a computational path of $M$ which ends with a rejecting configuration no matter which existential choices are made by P1, and $V$ rejects with probability 1 when this path is presented. Otherwise, P0 must sneak a transition error somewhere so that it can end up with a rejecting configuration. In any single simulation of $M$, $V$ will fail to catch such a transition error with probability at most $\frac{t-1}{t}$. The probability that a member of $L$ will be rejected by $V$ is thus ${ \left( \frac{t-1}{t} \right) }^d$, which can be shown to be not more than $\varepsilon$. It is clear that $V$ runs in polynomial time.

To prove $\mathsf{PSPACE} \subseteq \mathsf{ZDEB^*}(log, poly,poly)$, we use the same setup to simulate a polynomial-time blind alternating Turing machine in a zero-information debate checkable by a polynomial-time logspace verifier. In this case, P1 should not know about the universal choices of the simulated blind alternating machine. P0 guarantees this condition by using only the private alphabet $\Delta$ to talk to $V$.
\end{proof}

Thus, we have shown that allowing polynomial-time verifiers of complete-information debates to use polynomial amounts of randomness lets us constrain their space bounds logarithmically, without decreasing their power. Giving P0 further privacy does not add to the power under these bounds.
\begin{theorem}
\label{PSPACEeqDEB(log,poly)}
$\mathsf{CDEB}(log, poly,\infty) = \mathsf{ZDEB}(log, poly,\infty) = \mathsf{PDEB}(log, poly,\infty) = \mathsf{PSPACE}$.
\end{theorem}

\section{Nonapproximability of the quantified max word problem}
\label{sec:nonapp}

The \textit{quantified max word problem for matrices} (QMW problem) is defined as follows. Given a finite set $M$ of $m \times m$ matrices,  two $m$-length vectors $v$ and $w$, a bound $c$, and  quantifiers $Q_i \in \{ \exists, \forall\}$, is it possible to satisfy the inequality $Q_1M_1Q_2M_2...Q_kM_k [ vM_1M_2...M_kw^T>c ]$, where the $M_i$ variables will be selected from the members of $M$, for $1 \leq i \leq k$? 

The proof of Lemma \ref{PDEB(log,poly)inPSPACE} reveals that every language in $\mathsf{CDEB}(log, poly,\infty)$ is polynomial-time reducible to the QMW problem. Therefore, the QMW problem is $\mathsf{PSPACE}$-hard. The fact that any instance of the QMW problem can be solved using polynomial space by a simple exhaustive depth-first search on the finite game tree of the instance implies the following.

\begin{corollary}
The quantified max-word problem for matrices is $\mathsf{PSPACE}$-complete.
\end{corollary} 

We now consider the maximization version of the QMW problem, MAX-QMW. Suppose that the matrices in the inequality of an instance $\Pi$ of the QMW problem are chosen by two players named P1 and  P0. In particular, P0 and P1 choose matrices quantified by $\forall$ and $\exists$, respectively, in the order of quantification. Their game returns the result of the matrix multiplication $ vM_1M_2...M_kw^T$. Let $\Omega_\Pi$ be the maximum number that player P1 can guarantee to get as the product at the end, no matter what player P0 does. MAX-QMW is the function from the domain consisting of instances of the quantified max-word problem to their $\Omega$ values. 

We say that a function $g(x)$ can be approximated within factor $f(n) > 1$ if there is a polynomial-time algorithm which outputs a value in the interval $\left[ \frac{g(x)}{f(|x|)}, g(x)f(|x|) \right] $ for any $x$ in the domain of $g(x)$. 

We now state a corollary of Lemma \ref{PDEB(log,poly)inPSPACE}. It is shown in essentially the same way with Theorem 3.1 of \cite{Co93B}, which uses a result on one-way interactive proof systems that is similar to our Lemma \ref{PDEB(log,poly)inPSPACE}.   

\begin{corollary}
The maximization version of the QMW problem cannot be approximated within factor ${n^c}$ in polynomial time for any constant $c >0$, unless $\mathsf{P}=\mathsf{PSPACE}$.
\end{corollary}

\section{Concluding remarks}
\label{sec:conc}

One of the motivations for our model was to distinguish the three agents ($V$, P1, and P0) in the debate checking scenario clearly from each other. Interpretations of the alternating TM variants, as well as
some other models such as Condon's probabilistic game automata \cite{Co89}, sometimes merge the universal player and the verifier, and make it difficult to ask certain questions that are quite natural
in the three-person model. One such question is whether anything changes if we make the coins of the verifier public to the provers. In the model of \cite{Co89}, $V$'s coins are always visible to P0, and Condon shows that the class of languages with what we would call complete-information debates checkable with logspace verifiers whose coins are public to both P1 and P0 is contained in $\mathsf{NP}$. Our demonstration that $\mathsf{CDEB}(log, poly,poly) =  \mathsf{PSPACE}$ (Theorem \ref{PSPACEeqDEB(log,poly)}) therefore constitutes strong evidence that keeping the coins private increases the power of debate checkers.

We have seen that increasing the amount randomness available to the verifier enlarges the class of languages with complete-information debates, as demonstrated by the relations 
\[\mathsf{CDEB}(cons, \infty ,0)= \mathsf{REG}\subsetneq \mathsf{CDEB}(cons, \infty , cons) =  \mathsf{P},\]
and possibly by
\[\mathsf{CDEB}(log, poly , cons) = \mathsf{P}\subseteq \mathsf{CDEB}(log, poly,poly)= \mathsf{PSPACE}.\]
What is the effect of intermediate amounts of randomness, for example, can we characterize $\mathsf{CDEB}(log, poly , log)$?

When logspace polynomial-time verifiers on complete-information debates are allowed to use polynomially many coins, their language recognition power increases (Fact \ref{fact:ASPACE} and Theorem \ref{PSPACEeqDEB(log,poly)}), whereas similarly bounded verifiers on partial-information debates do not gain any additional power by randomness (Fact \ref{fact:pdeblogpoly0} and Theorem \ref{PSPACEeqDEB(log,poly)}). Since we do not know any characterization of $\mathsf{ZDEB}(log, poly, 0)$, whether randomness confers any benefit to logspace polynomial-time verifiers for zero-information debates is an open question.

\section*{Acknowledgement} We thank Abuzer Yakary{\i}lmaz, who introduced us to private alternation.

\appendix
\section{Equivalence of our model and private alternation}

\begin{lemma}
\label{ATMinDEB}
Given a PATM (resp. BATM) recognizing a language $L$, one can construct a deterministic verifier which checks partial-information (resp. zero-information) debates about membership in $L$, and has the same time and space bounds with the given PATM (resp. BATM).
\end{lemma}

\begin{proof}
Let $M$ be the given private alternating Turing machine. (The construction for blind alternating Turing machines is identical.) We assume without loss of generality that the set of universal states of $M$ is partitioned to two subsets $S_{pri}$ and $S_{pub}$, so that the machine can make moves that can change only the private (resp. public) portions of memory from states in $S_{pri}$ (resp. $S_{pub}$). We show how to construct a deterministic verifier $V$ that checks partial-information debates about $L$.

$V$ has the same state set as $M$. The universal and existential state sets of $M$ correspond to $R_0$ and $R_1$ (sets of states that read C0 and C1 to see which way to branch), respectively, in the verifier. $R_0$ is further partitioned to $R_{0,pri}$ and $R_{0,pub}$, corresponding to $S_{pri}$ and $S_{pub}$, respectively. The program of $V$ mimics that of $M$, but takes care of the following issue: States in $R_{0,pri}$ expect P0 to write a symbol from the private alphabet $\Delta$ in the reading cell, and jump to the accept state otherwise. Similarly, P0 is supposed to emit only members of the public alphabet  $\Gamma_0$ when $V$ is in a state in $R_{0,pub}$. This ensures that P1 knows only as much as it should about the configuration of $V$ at any point.


The existential player of $M$ has a winning strategy for this computation game \cite{Re79} if and only if there exists a well-formed debate subtree where all paths lead to an acceptance by $V$. Clearly, $V$ uses the same amount of memory and halts within the same number of steps with $M$. 
\end{proof}

\begin{lemma}
\label{DEBinATM}
Given a deterministic verifier that checks a partial-information (resp. zero-information) debate about membership in language $L$, one can construct a PATM (resp. BATM) which recognizes the same language within the  time and space bounds of the given verifier.
\end{lemma}

\begin{proof}
We build a private alternating Turing machine $M$ that simulates the given verifier $V$ by using universal and existential moves to select refuter and prover symbols, respectively, to feed to $V$. The important point is that $M$ takes all of its universal steps privately, and the simulation of $V$ is also performed on the private work tape. If a universal branching produces a symbol from the public alphabet $\Gamma_0$ between $V$ and P0, $M$ lets the existential player know about this by writing a copy of that symbol in the common memory area. This concludes the construction.
\end{proof}

\bibliographystyle{plain}
\bibliography{YakaryilmazSay}

\end{document}